\documentclass[lettersize,journal,twoside]{IEEEtran}
\pdfoutput=1
\usepackage{amsmath,amssymb,amsfonts}
\usepackage{algorithmic}
\usepackage{algorithm}
\usepackage{array}
\usepackage[caption=false,font=normalsize,labelfont=sf,textfont=sf]{subfig}
\usepackage{textcomp}
\usepackage{stfloats}
\usepackage{url}
\usepackage{verbatim}
\usepackage{graphicx}
\usepackage{cite}
\usepackage{subfig}
\usepackage[thmmarks,amsmath]{ntheorem}
\theoremseparator{:}

\newtheorem{lemma}{Lemma}
\newtheorem{proposition}{Proposition}
\theoremheaderfont{\it}
\theorembodyfont{\normalfont}
\theoremsymbol{\ensuremath{\Box}}
\newtheorem*{proof}{Proof}
\theoremheaderfont{\it}
\theorembodyfont{\normalfont}
\theoremsymbol{}
\newtheorem*{remark}{Remark}
\begin{document}

\title{Control-Oriented Power Allocation for Integrated Satellite-UAV Networks}

\author{Chengleyang Lei,~\IEEEmembership{Student Member,~IEEE}, Wei Feng,~\IEEEmembership{Senior Member,~IEEE}, Jue~Wang,~\IEEEmembership{Member,~IEEE,}
	Shi~Jin,~\IEEEmembership{Senior~Member,~IEEE,} and Ning Ge,~\IEEEmembership{Member,~IEEE}
\thanks{C. Lei, W. Feng, and N. Ge are with the Department of Electronic Engineering, Beijing National Research Center for Information Science and Technology, Tsinghua University, Beijing 100084, China (email: lcly21@mails.tsinghua.edu.cn, fengwei@tsinghua.edu.cn, gening@tsinghua.edu.cn).}
        \thanks{J. Wang is with the School of Information Science and Technology, Nantong University, Nantong 226019, China (email: wangjue@ntu.edu.cn).}
\thanks{S. Jin is with the National Mobile Communications Research Laboratory, Southeast University, Nanjing 210096, China (email: jinshi@seu.edu.cn).}
}




\maketitle

\begin{abstract}
This letter presents a sensing-communication-computing-control ($\textbf{SC}^3$) integrated satellite unmanned aerial vehicle (UAV) network, where the UAV is equipped with on-board sensors, mobile edge computing (MEC) servers, base stations and satellite communication module.  	
Like the \textit{nervous system}, this integrated network is capable of organizing multiple field robots in remote areas, so as to perform mission-critical tasks which are dangerous for human.	
Aiming at activating this \textit{nervous system} with multiple $\textbf{SC}^3$ loops, we present a control-oriented optimization problem. Different from traditional studies which mainly focused on communication metrics, we address the power allocation issue to minimize the sum linear quadratic regulator (LQR) control cost of all $\textbf{SC}^3$ loops. Specifically, we show the convexity of the formulated problem and reveal the relationship between optimal transmit power and intrinsic entropy rate of different $\textbf{SC}^3$ loops. For the assure-to-be-stable case, we derive a closed-form solution for ease of practical applications. After demonstrating the superiority of the control-oriented power allocation, we further highlight its difference with classic capacity-oriented water-filling method.         
\end{abstract}


\section{Introduction}
\IEEEPARstart{F}{ield} robots could perform mission-critical tasks that are dangerous for human. When being dispatched in remote areas without terrestrial cellular coverage, operation of the robots has to rely on non-terrestrial infrastructures, including satellites and unmanned aerial vehicles (UAVs)~\cite{satellite1,emergency3}. For such scenarios, a UAV platform needs to have integrated functionalities to support various requirements of robots, e.g., sensing, controlling, computing, and communication. For example, a UAV can be equipped with on-board sensors to collect scene information, with mobile edge computing (MEC) servers to analyze the situation and make quick decisions for robot control, with base stations to transmit control commands to robots, and with satellite communication module to support real-time communication to the remote cloud center~\cite{ref1}. This leads to a sensing-communication-computing-control ($\textbf{SC}^3$) integrated satellite-UAV network, in which efficient resource orchestration for all the related functionalities is important. 

In the $\textbf{SC}^3$ integrated networks, the MEC servers analyze the situation and compute future control actions according to data from sensors, then the UAV transmits control commands to the robots to guide their actions and transmits sensor data when necessary to the remote cloud center for advanced analysis. The whole process is performed in a closed-loop manner, which is referred as a $\textbf{SC}^3$ loop. Intuitively, a $\textbf{SC}^3$ loop can be regarded as a \textit{reflex arc}, with the network regarded as a 
\textit{nervous system}~\cite{nervous_system}, where multiple $\textbf{SC}^3$ loops would share and compete for resources. 

Existing studies on integrated satellite-UAV networks have mainly focused on communications. For example, Liu \textit{et al.} jointly optimized the channel allocation, power allocation, and hovering time of UAVs to maximize the data transmission efficiency~\cite{Liu2021}. Wang \textit{et al.} investigated a space-air-ground network and jointly optimized hovering altitude and power allocation, to maximize the network capacity~\cite{Wang2019}. However, in a $\textbf{SC}^3$ integrated satellite-UAV network, control and communication are closed coupled, and we will be more concerned with the control performance, i.e., the deviation of the control objective state from the desired state. Therefore, the control part should also be considered in the design of $\textbf{SC}^3$ integrated networks. 

Researchers in the control field have investigated the relationship between communication and control in $\textbf{SC}^3$ loops earlier. It was shown that a noisy linear control system can be stabilized (i.e., its state vector is bounded) only if the communication throughput in one control cycle exceeds the intrinsic entropy rate of the control system~\cite{control1}. Qiu \textit{et al.} further generalized this result to a multi-channel case~\cite{control2}. Recently, the lower bound of the minimum data rate to achieve a certain linear quadratic regulator (LQR) cost was presented, where LQR cost is a metric to measure the state deviation and energy consumption~\cite{LQR}. All these achievements have indicated that jointly optimizing control and communication is promising and significant in the $\textbf{SC}^3$ integrated satellite-UAV network. However, most of these works modeled communications as simple pipelines with simple parameters and left practical communication resource allocation undiscussed.

\IEEEpubidadjcol
Inspired by the efforts in the control field, some recent works have considered the control part as constraints in communication design. Chang \textit{et al.} maximized the spectral efficiency of a wireless control system subject to the control convergence rate constraint~\cite{Chang2019}.
Chen \textit{et al.} maximized the delay determinacy under the same constraint~\cite{Chen2021}. These studies have taken a great step towards $\textbf{SC}^3$ control-oriented optimization. However, they still focused on communication metrics. In mission-critical $\textbf{SC}^3$ integrated networks, the control performance may be more important and therefore should be treated as objective, rather than constraints.

In this work, we directly optimize the sum LQR control cost of all $\textbf{SC}^3$ loops, which is more essential to measure the control performance. Particularly, we establish a relation between the LQR cost and the communication data rate, and then formulate a control performance optimization problem, by optimizing the transmit power of multiple $\textbf{SC}^3$ loops. We show the convexity of the formulated problem, reveal the relationship between optimal transmit power and the intrinsic entropy rates of different $\textbf{SC}^3$ loops. For the assure-to-be-stable case, we also derive a closed-form solution for ease of practical applications. 
      
\section{System Model and Problem Formulation}
As shown in Fig. \ref{fig:system}, we consider a $\textbf{SC}^3$ integrated satellite-UAV network which serves multiple field robots for mission-critical tasks, such as handling nuclear materials. One satellite and one UAV (integrated with a sensor, a base station, a satellite communication module and a MEC server) jointly provide sensing, communication, and computing services for the robots. The network enables multiple $\textbf{SC}^3$ loops simultaneously. In each loop, the sensor senses the states of the object. The MEC analyzes the sensor data and correspondingly computes the control commands. Next, the UAV transmits the sensor data to satellite for advanced analysis when necessary, and meanwhile transmits the control commands to guide the robot to properly handle the control object. Our goal is to make this closed-loop control process fast and accurate.

\begin{figure} [t]
	\centering
	\includegraphics[width=0.95\linewidth]{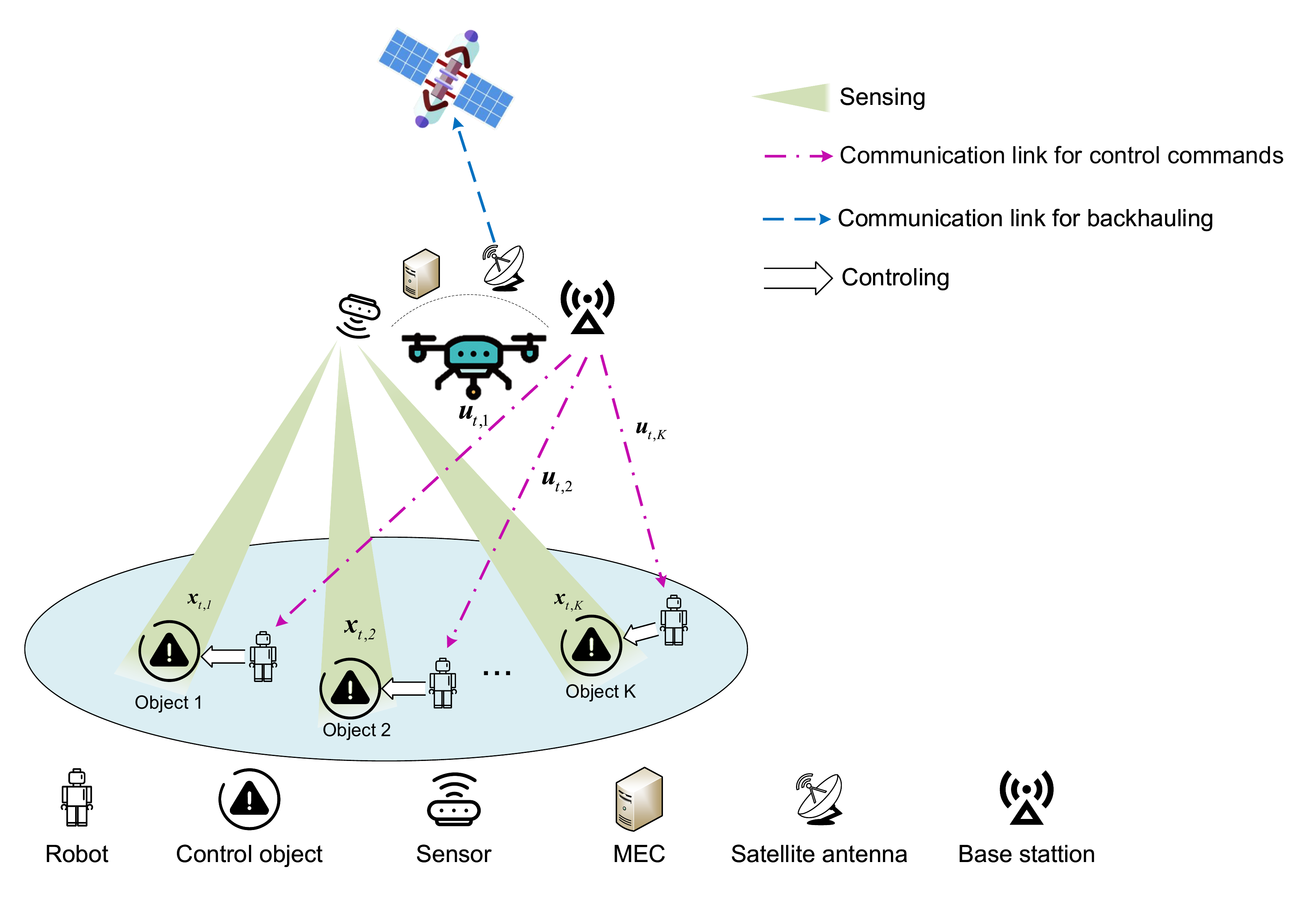}
	\caption{Illustration of a $\textbf{SC}^3$ integrated satellite-UAV network for organizing multiple field robots in remote areas.}
	\label{fig:system}
\end{figure}

The UAV simultaneously transmits the control commands to $K$ robots through orthogonal (e.g., in frequency) channels. Due to the limited transmit power, we have $\sum_{k = 1}^{K} p_k \leq P_{\text{max}}$, where $p_k$ denotes the power allocated to robot $k$ and $P_{\text{max}}$ represents the maximum transmit power of the UAV. The wireless channels between the UAV and robots are assumed to be dominated by line-of-sight (LoS) links~\cite{Zeng2016}. Therefore, the channel gain from the UAV to robot $k$ follows the free space path loss model as $g_k = \frac{\beta_0}{d^2_k}$, where $d_k$ denotes the distance from the UAV to robot $k$ and $\beta_0$ is the reference channel gain.

For the control part in $\textbf{SC}^3$ loops, we model each control object as a linear time-invariant system~\cite{control1, control2, LQR}, and hence the discrete-time system equation of the $k$th object is given by
\begin{equation}\label{system}
\mathbf{x}_{k,t+1} = \mathbf{A}_{k}\mathbf{x}_{k,t}+\mathbf{B}_{k}\mathbf{u}_{k,t}+\mathbf{v}_{k,t},
\end{equation}
where $t$ denotes the cycle index, $\mathbf{x}_{k,t}\in \mathbb{R}^{n}$ denotes the system state, such as the temperature or radiation intensity,  $\mathbf{u}_{k,t}\in \mathbb{R}^{m}$ denotes the control action, $\mathbf{v}_{k,t}\in \mathbb{R}^{n}$ denotes system noise, and $\mathbf{A}_{k}$ and $\mathbf{B}_{k}$ are fixed $n\times n$ and $n\times m$ matrices denoting the state matrix and input matrix respectively.

Due to the randomness of wireless channels between the UAV and robots, the communication data rate is limited, which may affect the control performance. According to \cite{control1}, to stabilize the control system $k$, the data throughput transmitted in each cycle needs to satisfy the condition
\begin{equation}\label{min_R}
BT_k\log_2\left( 1+\frac{g_kp_k}{\sigma^2}\right)  > h_k \triangleq \log_2 |\det \mathbf{A}_k|,
\end{equation}
where the left side of (\ref{min_R}) denotes the throughput in one cycle of channel $k$, $B$ is the bandwidth of each channel, $T_k$ is the time duration of each cycle in $\textbf{SC}^3$ loop $k$, $\sigma^2$ denotes the noise variance and $h_k$ is the intrinsic entropy rate which denotes the stability of object $k$. A large $h_k$ indicates an unstable control system, which requires high transmission rate to stabilize.

In control theory, the control performance can be measured by the LQR cost function. In this work, we consider the worst-case long-term average LQR cost, formulated as\cite{LQR}
\begin{equation}\label{LQR}
l_k \triangleq \sup\lim\limits_{N\rightarrow \infty}\mathbb{E} \left[ \sum_{t = 1}^{N} \left(\mathbf{x}_{k,t}^\text{T}\mathbf{Q}_{k}\mathbf{x}_{k,t} +\mathbf{u}_{k,t}^\text{T}\mathbf{R}_{k}\mathbf{u}_{k,t}\right) \right],
\end{equation}
where $\mathbf{Q}_k$ and $\mathbf{R}_k$ are semi-positive definite weight matrices. The term $\mathbf{x}_{k,t}^\text{T}\mathbf{Q}_{k}\mathbf{x}_{k,t}$ denotes the deviation of the system from zero state, and the term $\mathbf{u}_{k,t}^\text{T}\mathbf{R}_{k}\mathbf{u}_{k,t}$ denotes the control energy. These weight matrices balance the state and the energy, which can be set according to the practical requirements. For example, one should set the entries of $\mathbf{Q}_{k}$ to be large if he expects that the state of the system converges to zero quickly.

In order to achieve a certain LQR cost (denoted by $l_k$), the data throughput of channel $k$ in one cycle must satisfy the following constraint~\cite{LQR}
\begin{align}
BT_k\log_2\left( 1\!\!+\!\!\frac{g_kp_k}{\sigma^2}\right) \geq h_k \!\!+\!\! \frac{n}{2} &\log_2 \!\!\left(\! 1\!\!+\!\!\frac{ nN \!\left( \mathbf{v}_k\right)|\det \mathbf{M}_{k}|^\frac{1}{n} }{l_k-\text{tr}\left( \mathbf{\Sigma}_{k}\mathbf{S}_k\right)} \right),\nonumber\\
&\quad\quad\quad\ \  k = 1, 2, \cdots, K,\label{Rl}
\end{align}
where $N\!\left( \mathbf{v}_k\right) \triangleq \frac{1}{2\pi}\exp \left( \frac{2}{n} h\!\left(  \mathbf{v}_k \right)  \right) $ is the entropy power of the noise $\left( \mathbf{v}_k\right) $, $h\!\left(  \mathbf{v}_k \right)$ is its differential entropy, $\mathbf{\Sigma}_{k}$ is its covariance matrix, and $\mathbf{M}_{k}$ and $\mathbf{S}_k$ are the solutions to the following Riccati equations

\begin{subequations}\label{Riccati}
	\begin{align}
	&\mathbf{S}_k  = \mathbf{Q}_k + \mathbf{A}_k^\text{T}\left(\mathbf{S}_k- \mathbf{M}_k\right) \mathbf{A}_k,\\
	&\mathbf{M}_k  = \mathbf{S}_k \mathbf{B}_k \left( \mathbf{R}_k + \mathbf{B}^\text{T}_k \mathbf{S}_k \mathbf{B}_k\right) ^{-1} \mathbf{B}_k^\text{T} \mathbf{S}_k.
	\end{align}
\end{subequations}

In this work, we aim to minimize the sum long-term average LQR cost of the $\textbf{SC}^3$ loops by optimizing the power allocation $\mathbf{p}=\left[ p_1, p_2, \cdots, p_K \right] $, while keeping the communication constraint satisfied. The optimization problem is formulated as
\begin{subequations}\label{P1}
	\begin{align}
	\min_{\mathbf{p},\mathbf{l}} \quad&\sum_{k = 1}^{K} l_k  \label{P1a} \\ 
	\mbox{\textit{s.t.}}\quad & \sum_{k = 1}^{K} p_k \leq P_{\text{max}},\label{P1b}   
	\\
	&\text{(\ref{Rl})},  \label{P1d}
	\end{align}
\end{subequations}
where $\mathbf{l}=\left[ l_1, l_2, \cdots, l_K \right] $ and (\ref{P1d}) is the communication constraint imposed by the control performance requirements. In the next section, we will transform this problem to a convex problem and analyze the property of its optimal solution.

\section{Problem Transformation and Property Analysis}
We first propose a lemma to show the convexity of problem (\ref{P1}), accordingly, further reveal the relationship between the optimal power allocation and other parameters.

\begin{lemma}
Problem (\ref{P1}) is equivalent to the following convex problem
\begin{subequations}\label{P1'}
	\begin{align}
	\min_{\mathbf{p}} \quad&\sum_{k = 1}^{K} l_k \left( p_k \right) \label{P1a'} \\ 
	\mbox{s.t.}\quad & \sum_{k = 1}^{K} p_k \leq P_{\text{max}},  \label{P1b'}\\
	&BT_k\log_2\left( 1+\frac{g_kp_k}{\sigma^2}\right) > h_k, k = 1, 2, \cdots, K, \label{P1c'}
	\end{align}
\end{subequations}
where
\begin{align}\label{LQR_R}
l_k\left( p_k \right) \triangleq \frac{n 2^{\frac{2}{n}h_k} N \!\left( \mathbf{v}_k\right)|\det \mathbf{M}_{k}|^\frac{1}{n}} {{\left( 1+\frac{g_kp_k}{\sigma^2}\right)}^{\frac{2B T_k}{n} }-2^{\frac{2}{n}h_k}}+\text{tr}\left( \mathbf{\Sigma}_{k}\mathbf{S}_k\right).
\end{align}
\end{lemma}

\begin{proof}
As the right side of (\ref{Rl}) is monotonically decreasing with $l_k$, the equality must hold in order to minimize $l_k$ in the objective function. Otherwise, we can always reduce $l_k$ until that the equality is achieved. With (\ref{Rl}) replaced by equality, we obtain the equation between $l_k$ and $p_k$ as (\ref{LQR_R}). Correspondingly, problem (\ref{P1}) is recast as problem (\ref{P1'}), where constraint (\ref{P1c'}) ensures a positive denominator in (\ref{LQR_R}).

\begin{figure*}[!t]
	\begin{align}\label{derivative}
	\frac{\partial^2l_k}{\partial{p_k}^2} = \frac{2BT_k 2^{\frac{2}{n}h_k} N \!\left( \mathbf{v}_k\right)|\det \mathbf{M}_{k}|^\frac{1}{n}\left( 1 \!+\! \frac{g_kp_k}{\sigma^2}\right) ^{\frac{2BT_k}{n}}\left[ \frac{2BT_k}{n}2^{\frac{2}{n}h_k}\! +\! \frac{2BT_k}{n}\left( 1 \!+\! \frac{g_kp_k}{\sigma^2}\right) ^{\frac{2BT_k}{n}} \!+\! \left( 1 \!+\! \frac{g_kp_k}{\sigma^2}\right) ^{\frac{2BT_k}{n}} \!-\! 2^{\frac{2}{n}h_k}  \right] }{\left( p_k+\frac{g_k}{\sigma^2}\right) ^2\left[ \left( 1+\frac{g_kp_k}{\sigma^2}\right) ^{\frac{2BT_k}{n}}-2^{\frac{2}{n}h_k}\right] ^3}.
	\end{align}
	\hrulefill
\end{figure*}

Next, we prove that problem (\ref{P1'}) is convex. The second order derivative of $l_k\left( p_k \right) $ with respect to $p_k$ is shown as (\ref{derivative}). From (\ref{derivative}), we can find that $\frac{\partial^2l_k}{\partial{p_k}^2}>0$ as long as  $\left( 1+\frac{g_kp_k}{\sigma^2}\right) ^{\frac{2BT}{n}}>2^{\frac{2}{n}h_k}$, which is equivalent to (\ref{P1c'}). Therefore, the objective function is convex in its feasible region. In addition, it is easy to show that the feasible region is affine, which guarantees the convexity of (\ref{P1'}).
\end{proof}

\textbf{Lemma 1} shows that problem (\ref{P1}) can be transformed to a convex problem, whose optimal solution can be obtained efficiently. In the following, we first consider the special case that the LQR weight matrices are set as $\mathbf{Q}_{k} = \mathbf{I}_{n}$ and $\mathbf{R}_{k} = \mathbf{0}$. With this setting, the LQR cost solely describes the system deviation (from zero state) while the energy cost is not the focus. For this case, the relation between optimal power allocation and system stability is discussed in \textbf{Proposition 1}.

\begin{proposition}
	When the control energy cost is not concerned, the optimal power $p_k^*$ allocated to channel $k$ is monotonically increasing with respect to the intrinsic entropy rate $h_k$, i.e., the more unstable the control system $k$, the more power should be allocated to it.
\end{proposition}

\begin{proof}
When $\mathbf{Q}_{k} = \mathbf{I}_{n}$ and $\mathbf{R}_{k} = \mathbf{0}$, the solutions of the Riccati equations in (\ref{Riccati}) are $\mathbf{S}_k = \mathbf{M}_k = \mathbf{I}_n$~\cite{LQR}. Therefore, we have $|\det \mathbf{M}_k|^{\left( \frac{1}{n}\right)}  = 1$ in (\ref{LQR_R}).

It is not difficult to verify that the Slater's conditions hold for problem (\ref{P1'}), which guarantees strong duality~\cite{cvx}. Therefore, (\ref{P1'}) is equivalent to its Lagrangian dual problem
\begin{subequations}\label{P2}
	\begin{align}
	\max_{\lambda} \min_{\mathbf{p}} \quad&\sum_{k = 1}^{K} l_k \left( p_k \right) +\lambda \left( \sum_{k = 1}^{K} p_k-P_{\text{max}} \right) \label{P2a} \\ 
	\mbox{\textit{s.t.}}\quad & \lambda\geq 0  \label{P2b},\\
	&\left( \text{\ref{P1c'}} \right) , \label{P2c}
	\end{align}
\end{subequations}
where $\lambda$ is the Lagrangian multiplier with respect to constraint (\ref{P1b'}). By checking the Karush-Kuhn-Tucker (KKT) conditions, the optimal solution to problem (\ref{P1}), denoted as $\left\{p^*_k\right\}$, must satisfy the following equations
\begin{subequations}\label{P3}
	\begin{align}
	&\frac{\partial l_k}{\partial{p^*_k}}+\lambda = 0, k = 1,2,\cdots, K, \label{P3a} \\ 
	& \lambda\geq 0,  \label{P3b}\\
	&BT_k\log_2\left( 1+\frac{g_k p^*_k}{\sigma^2}\right) > h_k, k = 1, 2, \cdots, K, \label{P3c}\\
	&\sum_{k = 1}^{K} p^*_k = P_{\text{max}}, \label{P3d}
	\end{align}
\end{subequations}
where (\ref{P3d}) follows from the monotonicity of $l_k$ with respect to $p_k$. 

Calculating $\frac{\partial l_k}{\partial{p^*_k}}$ in (\ref{P3a}), we obtain the following equation
\begin{equation}\label{P3a1}
\frac{2BT_k 2^{\frac{2}{n}h_k} N \!\left( \mathbf{v}_k\right) g_k \left( 1 + \frac{g_k p^*_k}{\sigma^2}\right) ^{\frac{2BT_k}{n} - 1}}{\sigma^2\left[  \left( 1 + \frac{g_k p^*_k}{\sigma^2}\right)^{\frac{2BT_k}{n}}-2^{\frac{2}{n}h_k}\right] ^2} = \lambda.
\end{equation}

Denoting the left side of (\ref{P3a1}) as $s_k\left( h_k, p_k\right)$ , we have
\begin{equation}\label{P3a2}
s_k \left(h_k, p_k\right)  = s_j \left( h_j, p_j\right),\quad \forall j,k = 1, 2, \cdots, K.
\end{equation}
	
Similar to (\ref{derivative}), we have $\frac{\partial s_k}{\partial{p_k}}<0$, which means that $s_k$ is monotonically decreasing with $p_k$. In addition, $s_k$ is monotonically increasing with $h_k$. Therefore, when the other parameters are given and fixed, if $h_k$ increases, $p_k$ should also increase to guarantee that the condition (\ref{P3a2}) still holds.
\end{proof}

\begin{remark}
Similarly, we can prove that the optimal power allocated to $\textbf{SC}^3$ loop $k$ is monotonically increasing with the entropy power of $\mathbf{v}_k$. Therefore, we can draw a conclusion that one should allocate more power to the unstable control systems with larger noise to improve the overall control performance.
\end{remark}

Next, we derive closed-form expression of the optimal power allocation without restricting the forms of $\mathbf{Q}_k$ and $\mathbf{R}_k$. For simplicity, we consider the assure-to-be-stable assumption, that the communication capability of each control system is significantly greater than the lowest capability requirement to keep the system stable in (\ref{min_R}), i.e., $BT_k\log_2\left( 1+\frac{g_k p_k}{\sigma^2}\right) \gg h_k$. This assumption means that the control systems are far from the unstable point, and hence we can focus on the control performance instead of the stability.

\begin{proposition}
	Under the assured-to-be-stable assumption, if all of the $\textbf{SC}^3$ loops have the same cycle time, i.e., $T_1 = T_2 = \cdots = T_K = T$, the optimal solution to problem (\ref{P1}) is obtained as (\ref{pk*1}).
	\begin{figure*}[!t]
		\begin{equation}\label{pk*1}
		p^*_k = \frac{\left(P_{\text{max}}+\sum_{k = 1}^{K}\frac{\sigma^2}{g_k} \right)  {\left[|\det \mathbf{M}_{k}|^\frac{1}{n}  2^{\frac{2}{n}h_k}  N \!\left( \mathbf{v}_k\right)\right] } ^ {\frac{n}{{2BT}+{n}}} \left( \frac{\sigma^2}{g_k}\right) ^{\frac{2BT}{{2BT}+{n}}}}{\sum_{k = 1}^K  {\left[|\det \mathbf{M}_{k}|^\frac{1}{n}  2^{\frac{2}{n}h_k}  N \!\left( \mathbf{v}_k\right)\right] } ^ {\frac{n}{{2BT}+{n}}} \left( \frac{\sigma^2}{g_k}\right) ^{\frac{2BT}{2BT+{n}}}}-\frac{\sigma^2}{g_k}.
		\end{equation}
	\end{figure*}
\end{proposition}

\begin{proof}
	When $BT_k\log_2\left( 1+\frac{g_k p_k}{\sigma^2}\right) \gg h_k$, we have $\left( 1 + \frac{g_k p^*_k}{\sigma^2}\right)^{\frac{2BT}{n}} \gg 2^{\frac{2}{n}h_k}$, which means that the term $2^{\frac{2}{n}h_k}$ in the denominator of the left hand in (\ref{P3a1}) is negligible. Therefore, we can rewrite (\ref{P3a1}) as
	\begin{equation}\label{P3a3}
	\frac{2BT |\det \mathbf{M}_{k}|^\frac{1}{n}  2^{\frac{2}{n}h_k}  N \!\left( \mathbf{v}_K\right) g_k }{\sigma^2 \left( 1 + \frac{g_k p^*_k}{\sigma^2}\right)^{\frac{2BT}{n}+1}} = \lambda.
	\end{equation}
	
	From (\ref{P3a3}), we have\footnote{The assumption $BT\log_2\left( 1+\frac{g_k p_k}{\sigma^2}\right) \gg h_k$ guarantees that the power is greater than zero, so we don't need the operator $\left( \cdot\right) ^+$ in (\ref{pk*}).}
	\begin{equation}\label{pk*}
	p^*_k = \left[ \left( \frac{2BT |\det \mathbf{M}_{k}|^\frac{1}{n}  2^{\frac{2}{n}h_k}  N \!\left( \mathbf{v}_K\right) g_k }{\lambda \sigma^2} \right) ^{\frac{n}{{2BT}+{n} }}-1\right] \frac{\sigma^2}{g_k}.
	\end{equation}
	Based on (\ref{pk*}) and (\ref{P3d}), we obtain that the Lagrangian multiplier $\lambda$ satisfies the equation (\ref{lambda}).
	\begin{figure*}[!t]
		\begin{equation}\label{lambda}
		\left( \frac{1}{\lambda}\right)  ^{\frac{n}{{2BT}+{n} }} =  \frac{P_{\text{max}}+\sum_{k = 1}^{K}\frac{\sigma^2}{g_k}}{\sum_{k = 1}^K \left( 2BT |\det \mathbf{M}_{k}|^\frac{1}{n}  2^{\frac{2}{n}h_k}  N \!\left( \mathbf{v}_K\right) \right) ^ {\frac{n}{2BT+n}} \left( \frac{\sigma^2}{g_k}\right) ^{\frac{2BT}{2BT+{n}}}}.
		\end{equation}
		\hrulefill
	\end{figure*}
	Substituting (\ref{lambda}) into (\ref{pk*}), we can obtain (\ref{pk*1}) immediately, which completes the proof.
\end{proof}

\begin{remark}
	From (\ref{pk*1}), we see that $p_k^*$ is increasing with $h_k$ and $N\left( \mathbf{v}_k\right) $, which verifies the conclusion of \textbf{Proposition 1}. In addition, we can see the optimal power allocation $p^*_k$ is exponentially increasing with the intrinsic entropy rate $h_k$, and is polynomially decreasing with $g_k$, which means that the intrinsic entropy parameter may be more important than the channel gain parameter and should be paid more attention to in the power allocation.
\end{remark}

\section{Simulation Results}
In this section, we provide simulation results to  demonstrate the performance of the control-oriented power allocation method, and verify our previously derived conclusions.

We assume that there are $5$ objects to be controlled, all randomly and evenly distributed in a circular area with a radius of $5000$ m. The UAV is deployed at the center of the circle with the height of $1000$m. The bandwidth of each sub-channel is $5$kHz unless otherwise specified. Other parameters are set as $\beta_0 = -60$dB and $\sigma^2  =-110$dBm~\cite{Hua2018}.

For the control part, unless otherwise specified, the intrinsic entropy rates of each system are randomly selected from the range $\left[ 0, 100 \right]$, the system noise is assumed to be independent Gaussian random variables with mean zero and variance 0.01, and the other parameters are set as $T_1 = T_2 = \cdots  =T_K = 10$ms and $n = 100$. The LQR weight matrices are $\mathbf{Q}_k =  \mathbf{I}_n, \mathbf{R}_k = \mathbf{0}, \quad \forall k = 1,2,\cdots, K$. 

We compare our scheme with the traditional water-filling
power allocation, which is proved to be optimal for maximizing the sum rate capacity. The optimal transmit power with water-filling allocation can be calculated as ~\cite[Chapter 9.4]{Elements}
\begin{equation}\label{water}
p_k^\text{WF} = \left( \frac{1}{\lambda_0} - \frac{\sigma^2}{g_k} \right) ^+,
\end{equation}
where $p_k^\text{WF}$ denotes the power allocated to channel $k$, $\lambda_0$ is chosen to satisfy $\sum_{k=1}^{K}p_k^\text{WF} = P_\text{max}$, and $\left( x\right) ^+ = \max \left\{x,0\right\}$. 

\begin{figure} [t]
	\centering
	\includegraphics[width=\linewidth]{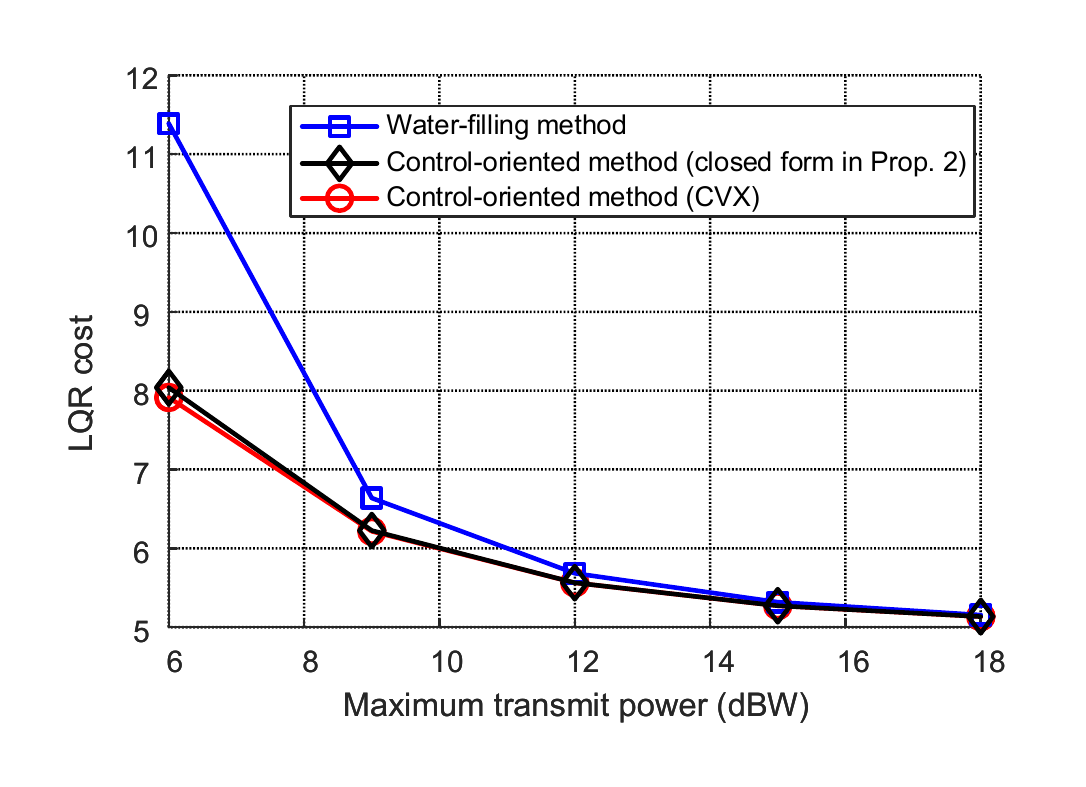}
	\caption{LQR cost achieved by different power allocation methods.}
	\label{simu1}
\end{figure}

In Fig. \ref{simu1}, we compare the LQR costs achieved by the water-filling method, the closed-form power allocation in \textbf{Proposition 2} and the optimal results to (\ref{P1}) obtained with CVX. The figure shows the accuracy of the approximate expression in (\ref{pk*1}), as the optimal solution obtained by CVX and the closed-form in \textbf{Proposition 2} achieve nearly the same LQR costs. When the maximum power is $6$dBW, the LQR cost of the power allocation in \textbf{Proposition 2} is slightly higher than that of the optimal solution, because the assumption $BT\log_2\left( 1+\frac{g_k p_k}{\sigma^2}\right) \gg h_k$ is not satisfied. From this figure, we can see that the LQR cost decreases with the maximum power. This is because the robot can receive more accurate control commands with more transmit power. In addition, the LQR cost with both control-oriented methods is lower than that with conventional water-filling, which verifies the superiority of the proposed method. Notably, when the maximum power becomes large enough, the LQR costs obtained by the three methods tend to be close. This is because the communication ability significantly exceeds control requirements and the LQR cost is near to its ideal value, i.e., $\text{tr}\left( \mathbf{\Sigma}_{k}\mathbf{S}_k\right)$ as in (\ref{LQR_R}).

\begin{figure} [t]
	\centering
	\includegraphics[width=0.96\linewidth]{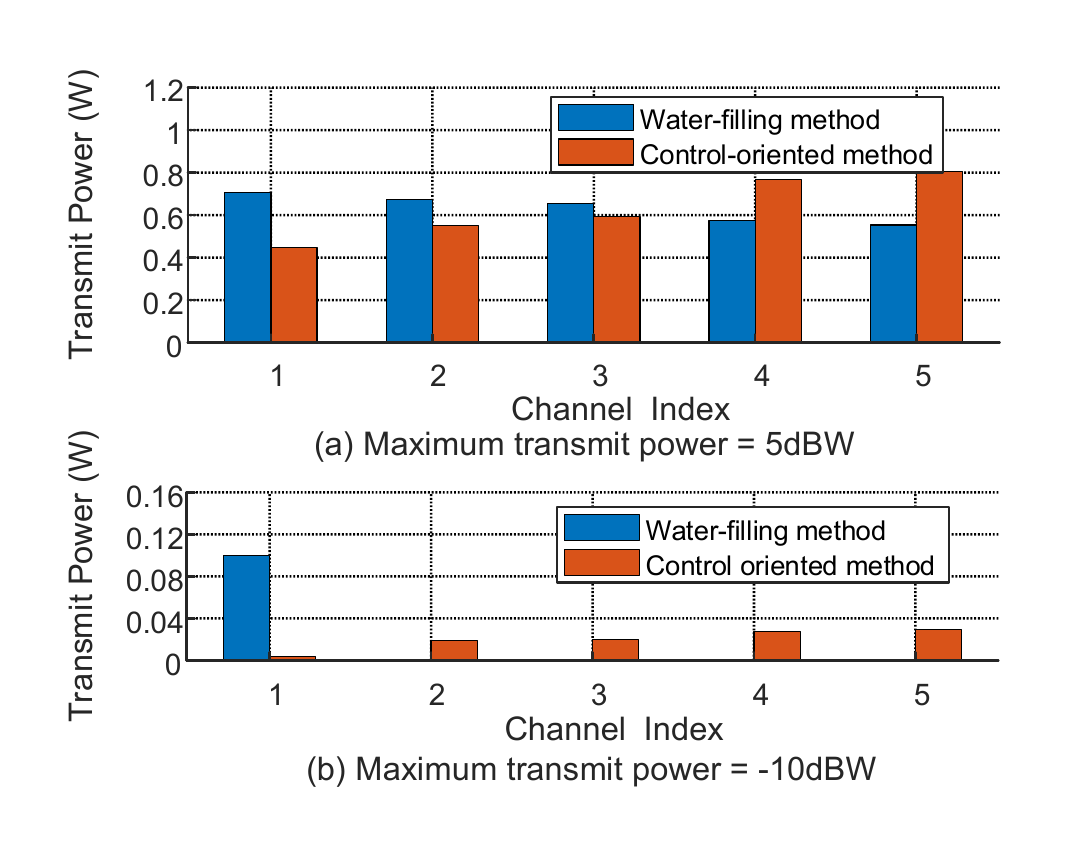}
	\caption{Power allocated to channels with different channel gains by different methods under different maximum transmit power constraints.}
	\label{simu2_3}
\end{figure}

Fig. \ref{simu2_3}a compares the power allocated to each channel, obtained with the water-filling and control-oriented methods, respectively, where the maximum power is set as $5$dBW. To highlight the influence of the channel conditions, the intrinsic entropy rates of each system are set to be $5$. The channels are sorted such that $g_1\geq g_2 \geq \cdots \geq g_5$. In Fig. \ref{simu2_3}a, difference between these two methods is clearly shown. The control-oriented allocation method tends to allocate more power to the channels with bad channel conditions, while the water-filling method behaves oppositely. This conclusion can be derived by comparing the expressions in (\ref{pk*1}) and (\ref{water}). To show the difference more clearly, we further compare the power allocation results for an extreme condition with a much lower maximum power in Fig. \ref{simu2_3}b, where $P_\text{max} = -10$dBW. It is seen that the water-filling method in this case allocates all power to the channel with the best condition, while the control-oriented allocation method still allocates power to every channel to ensure the control system stability under constraint (\ref{P1c'}).

 \section{Conclusions}
In this letter, we investigated a $\textbf{SC}^3$ integrated satellite-UAV network. A control-oriented power allocation problem was formulated. We transformed it into a convex problem and proved that more power should be allocated to the loop with higher intrinsic entropy rate so as to improve the control performance. We further derived the closed-form expression of the optimal power in the assure-to-be-stable case. Simulation results showed the big difference between the control-oriented method and the conventional communication-oriented water-filling method. Specifically, the former will allocate more power to the channel with worse conditions while the later behaves oppositely.


\newpage

 


\vspace{11pt}


\vfill

\end{document}